\documentclass[a4paper,11pt]{article}
\usepackage{url}

\usepackage{amsmath, amssymb}
\usepackage{amsthm}
\usepackage[OT1]{fontenc}
\usepackage[colorinlistoftodos]{todonotes}

\usepackage{sectsty,textcase}
\usepackage{chngcntr}
\counterwithin{figure}{section}

\usepackage[margin=1.1in]{geometry}

\newtheorem{thm}{Theorem}[section]
\newtheorem{definition}[thm]{Definition}
\newtheorem{lemma}[thm]{Lemma}
\newtheorem{assumption}[thm]{Assumption}
\newtheorem{rmk}[thm]{Remark}
\newtheorem{emp}[thm]{Example}
\newtheorem{cmt}[thm]{Comment}
\newtheorem{cor}[thm]{Corollary}
\newtheorem{proposition}[thm]{Proposition}

\newcommand{\rr}{\rightarrow}
\newcommand{\comment}[1]{}
\newcommand{\td}[1]{\tilde{#1}}
\newcommand{\dl}[1]{\underline{#1}}
\newcommand{\ul}[1]{\bar{#1}}
\newcommand{\1}{\mathbf{1}}
\newcommand{\Prob}{\mathbb{P}}
\newcommand{\E}{\mathbb{E}}
\newcommand{\G}{\mathbb{G}}
\newcommand{\Gt}{\mathcal{G}}
\newcommand{\F}{\mathbb{F}}
\newcommand{\Ft}{\mathcal{F}}

\usepackage{xcolor}

\usepackage{lineno}

\title{Intensity Process for a Pure Jump L\'evy Structural Model with Incomplete Information}
\author{Xin Dong,  Harry Zheng\footnote{Department of Mathematics, Imperial
College, London SW7 2AZ, UK. 
 E-mail addresses: x.dong10@imperial.ac.uk, h.zheng@imperial.ac.uk.}\\
 \smallskip Imperial College}
\date{}

\begin{document}

\maketitle

\bigskip\noindent{\bf Abstract} \
In this paper  we discuss a credit risk  model with a pure jump L\'evy process for the asset value and an unobservable random barrier. The default time is the first  time when the asset value falls below the barrier.  Using the  indistinguishability of  the intensity process and the likelihood process, we prove the existence of the intensity process of the default time and  find its explicit representation in terms of the distance between the asset value and its running minimal value. We apply the result to find the instantaneous credit spread  process  and illustrate it with a numerical example.




\section{Introduction}
The structural model and the intensity model are two frameworks in credit
risk modelling. The structural model is based on the asset-liability structure
of a firm and is economically meaningful. The default time is defined as
the first time the asset value process falls below the default threshold.
One needs to investigate the law of the first passage time or equivalently
the running minimal process. The intensity form model is based on the the fact
that default happens as a surprise to the market and default time is a totally
inaccessible stopping time under a certain filtration. One models  directly the intensity
process that determines the default
indicator process and  the short-term  spread of credit derivatives such
as defaultable bonds and credit default swaps. 

The key difference of the two models is the difference of the information
sets or filtrations, see Jarrow and Protter \cite{PJ04}. If the asset value
process is continuous  and the barrier is deterministic in a structural model
with complete information, then the first passage time is a predictable stopping
time and  does not admit an intensity process under the natural filtration.
In reality, it is difficult to observe the complete information of the asset
value process and the default barrier. 
 There has been active research in the literature on the 
 filtration expansion and its applications in the structural model with incomplete
information, see Guo and Zeng \cite{GZ08} and Janson et al. \cite{JSP11}.

There are two main ways of introducing the incomplete information in the first passage
time model in the literature. One is to assume
the incomplete information about the value process and the constant barrier.
Duffie and Lando \cite{DuffieLando03} discuss for a discretely observable noisy value
process and find the corresponding intensity process . Kusuoka \cite{Kusuoka99} extends \cite{DuffieLando03} to a continuously observable noisy value process. \c{C}etin et al.~\cite{Cetin04}
derive the intensity process with the A\'ezma martingale and the information reduction method.
The other is to assume the observable asset value process but the incomplete
information on the random barrier. Giesecke \cite{Giesecke06}  introduces
 an  unobservable random  barrier and concludes that if the asset value
is a diffusion process then the  default time  is a totally inaccessible
stopping time under the market information filtration  but does not admit
an intensity process. 

In this paper we focus on the first passage time problem of a structural
model for a L\'evy process with finite variation and with incomplete information
of the barrier.
Pure jump processes are important in financial modelling as they can capture
the phenomenon of infinite activities, jumps, skewness and kurtosis. For
example, Madan et al.~\cite{MCC98} use a variance gamma
process for the stock price in option pricing. Madan and Schoutens \cite{MadanSchoutens08} use
a drifted subordinator for the log firm value process in a first passage
time model with complete information. 

In the incomplete information setup the essential mathematical quantity needed
is the conditional default probability. All results in the literature on
the existence of the intensity process are based on the absolute continuity
of the conditional default probability and the close link between the conditional
default density and the intensity. In case of  pure jump processes the conditional
default probability is discontinuous at the time when the asset value process
reaches a new minimal and the conditional default density does not exist. This is
reasonable as one would expect the conditional default probability jumps  when
there is a large  movement of the asset value process.
The main mathematical difficulty, unlike the continuous case in which the
compensator of the conditional default probability is itself, is to find the
compensator due to the unpredictability of the stopping time. 

The objective of the paper is to show that the structural model of a pure jump L\'evy
process  with an unobservable random barrier can be embedded into an equivalent intensity model. 
The key contribution of the paper is to show the existence of the intensity
process and find its explicit form  for a pure jump  L\'evy  process in an
incomplete information framework, which  sheds the new light to the relation
between the intensity process of the default time and the running minimal
process of the asset value. We apply the result to find the instantaneous credit spread  process  
that remains positive and finite, which conforms to the market observations, and that depends on the historical path of the asset value. 

The paper is organized as follows. Section~\ref{sec: model and main} introduces
the model,  states the main result  (Theorem \ref{thm: main}) with several
examples, and discusses the instantaneous credit spread with a numerical example.  
Section~\ref{sec: proof main theorem} proves the main result with details
discussed in four subsections. 
Section~\ref{sec: conclusion} concludes.

\section{The Model and  the Main Result}\label{sec: model and main}
 Let $(\Omega,\mathcal{G},\mathbb{P})$ be a probability space and $V$ be
an observable firm asset value  process given by $V_t = V_0 e^{X_t}$ at time $t$, 
where $X$ is a  L\'evy process with finite variation and $X_0 = 0$. Examples
include drifted subordinators, variance gamma and normal inverse Gaussian
processes. Note that $X$ can be decomposed as (\cite[Exercise 2.8]{Kyprianou06})
\begin{eqnarray}\label{eqn: decompFVLP}
X_t = ct - S_t + S'_t,
\end{eqnarray}
where $c \in \mathbb{R}$ and $S$, $S'$ are independent pure jump subordinators
with L\'evy measures $\pi$, $\pi'$, respectively, see \cite[Lemma 2.14]{Kyprianou06}
for the definition and the properties of a subordinator. Denote by $\mathbb{F}=(\mathcal{F}_t)$
the 
natural filtration generated by $X$. We assume the following assumption be
satisfied in the paper:
\begin{assumption}\label{asm: negLevyMeasure}
L\'evy measure $\pi$ is continuous and satisfies
 $\int_0^\infty x \pi(dx) < \infty$. 
\end{assumption}

Assume that the firm defaults at the first time when the asset value is below
a default threshold, i.e., 
the default time $\tau$ is defined by
\begin{eqnarray*}
\tau:=  \inf\{t >0: V_t\leq  \tilde D\}=\inf\{t >0: X_t \leq  D:=\ln(\tilde
D/V_0)\},
\end{eqnarray*} 
where $\tilde D$ is an unobservable default barrier of the company. 
Assume that 
$\tilde{D}$ is a uniform 
 variable on the interval $[0, V_0]$ and is independent of $V$.  Then  the
barrier $D$ for $X$ is a standard negative exponential  variable, i.e., $-D$
is a standard exponential variable, with the distribution function $ \Prob(D\leq
x) = e^x$ for $x <0$, and is independent of $X$. Note that the default barrier is unobservable but the default time is observable, we therefore define a  progressive filtration expansion $\G = (\Gt_t)$ 
by (\cite[Chapter VI, Section 3]{Protter05})

\begin{eqnarray}\label{eqn:progEnlargement}
\mathcal{G}_t = \{B \in \mathcal{G}: \exists B_t \in \mathcal{F}_t, B \cap
\{\tau > t\} = B_t \cap  \{\tau > t\}\}.
\end{eqnarray}
The default time  $\tau$ is now  a $\G$-stopping time.
All filtrations involved are assumed to satisfy the usual condition.  

Denote  by $N$ the default indicator process, defined by $N_t :=  \1_{\{\tau \leq t\}}$.
The Doob-Meyer decomposition theorem implies that there exists a unique increasing
predictable process $A$ with $A_0=0$, called  the $\G$-compensator of $N$,
such that $N -A$ is a $\G$-martingale. If $A$ is continuous a.s. then $\G$-stopping
time $\tau$ is totally inaccessible.
If $A$ is absolutely continuous a.s. with respect to the Lebesgue measure
and $A$ can be written as $A_t = \int_0^t \lambda_s ds$ a.s., where $\lambda$
is nonnegative and $\G$-progressively measurable, then $\lambda$ is called
the intensity process of $N$, see \cite{Bremaud81} for details on  compensators
and intensity processes.

Denote by $\pi(x+du) := \pi((x+u,x+u+du])$.
If $\pi$ admits a L\'evy density $\nu$, then $\pi(x+du)=\nu(x+u)dx$.
We can now state the main result of the paper.

\begin{thm}\label{thm: main}
Let $X$ be a L\'evy process with finite variation and  Assumption~\ref{asm: negLevyMeasure} be satisfied.  Then the
$\G$-compensator
of the default indicator process $N$ is absolutely continuous a.s. and 
 the intensity process $\lambda$ of $N$ is indistinguishable with the instantaneous
likelihood process $\td{\lambda}_t := \lim_{h\downarrow 0}\frac{1}{h}\Prob(t<
\tau \leq t+h|\mathcal{G}_t)$ on $\{\tau>t\}$.
Moreover,  using the same notation
as in  (\ref{eqn: decompFVLP}), the intensity process $\lambda$ has the following representation
\begin{eqnarray}\label{eqn: lambda_t}
\lambda_t =\mathbf{1}_{\{\tau>t\}}\left(  -c\mathbf{1}_{\{X_t-\dl{X}_t=0\}}\mathbf{1}_{\{c<0\}}+\Pi(X_t-\dl{X}_t)\right),
\end{eqnarray}
where  $\dl{X}_t:=
\inf_{0\leq s \leq t} X_s$ is the running minimal process of $X$ and
\begin{equation} \label{Pi}
\Pi(x) := \int_{0}^\infty (1-e^{-u})\pi(x+du), \quad \forall x\geq 0.
\end{equation}
\end{thm}

\comment{
\begin{thm}
If $X$ is a pure jump L\'evy process with infinite variation, then the compensator
of $\tau$ under $\G$ is continuous which indicates that $\tau$ is totally
inaccessible.
\end{thm}
}

Theorem~\ref{thm: main} shows that the intensity process $\lambda$ is an
endogenous process that depends on the path of the asset value
 process $X$. Moreover, at each time $t$, $\lambda_t$ is a decreasing function of $X_t-\dl{X}_t$,
 a financially desirable property as it means that the default intensity
increases when the asset value process $X$ approaches its historical minimal
level.

  We next give several examples to illustrate Theorem \ref{thm: main}.

\begin{emp} (Drifted Compound Poisson Process)
{\rm Let $X$ be given by
\begin{eqnarray*} 
X_t = ct  - \sum_{i=1}^{M_t}Y_i + \sum_{i=1}^{M'_t}Y'_i,
\end{eqnarray*}
where $c\in \mathbb{R}$,  $Y_i$ and $Y'_i$ are exponential variables with
parameters $\beta$ and $\beta'$, respectively, $M$ and $M'$ are Poisson
processes with intensities $\rho$ and $\rho'$, respectively, and $\{Y_i\}$,
$\{Y'_i\}$, $M$, $M'$ are independent of each other. The L\'evy density
 of $X_t$ on $\mathbb{R}_{-}$ is given by $\nu_{-}(x)= \rho \beta e^{-\beta x}$. The
intensity process $\lambda$ of  the default indicator process $N$ is then
given by Theorem~\ref{thm: main} as   
\begin{eqnarray*}
\lambda_t &=& \mathbf{1}_{\{\tau>t\}}\left(
 -c\mathbf{1}_{\{X_t-\dl{X}_t=0\}}\mathbf{1}_{\{c<0\}}+  \frac{\rho}{1+\beta}
e^{-\beta(X_t - \dl{X}_t)} \right).
\end{eqnarray*}
}
\end{emp}

\begin{emp} (Drifted Gamma Process)\label{emp: DG}
{\rm Let $X$ be given by
\begin{eqnarray*}
X_t =ct -G_t,
\end{eqnarray*}
where $c>0$, $G_t$ is a gamma process $\Gamma(t,\mu,\nu)$ with the mean rate $\mu$, the variance rate $\nu$, and  the L\'evy density
$\nu(x)=  \frac{\mu^2}{\nu} e^{-\frac{\mu}{\nu}x}x^{-1}$.  The intensity process  of  $N$ is given by
\begin{eqnarray*}
\lambda_t = \mathbf{1}_{\{\tau>t\}}\left( \int_{0}^\infty (1-e^{-u}) \frac{\mu^2}{\nu}e^{-\frac{\mu}{\nu}(u+X_t
-\dl{X}_t)}(u+X_t-\dl{X}_t)^{-1}du \right).
\end{eqnarray*}
Note that $c>0$ in this case, hence the first term in (\ref {eqn: lambda_t})  disappears.}
\end{emp}

\begin{emp} (Variance Gamma Process \cite{MCC98})\label{emp: VG} {\rm
Let $X$ be a variance gamma process $VG(c,\nu,\sigma,\theta)$ that is generated by a drifted Brownian
motion $\theta t + \sigma W_t$, time-changed by a gamma process $\Gamma(t;1,\nu)$, and an additional drift term $ct$, then 
\begin{equation}\label{eqn: VG_decomposition}
X_t = ct+ \Gamma(t;\mu_+,\nu_+)-\Gamma(t;\mu_-,\nu_-),
\end{equation}
where $\mu_{\pm} = \frac{1}{2}\sqrt{\theta^2+\frac{2\sigma^2}{\nu}}\pm
\frac{\theta}{2}$, and $\nu_{\pm}= \mu_{\pm}^2 \nu$. 
The intensity process of $N$ is given by
\begin{eqnarray}\label{eqn: lambda_VG}
\lambda_t =\mathbf{1}_{\{\tau>t\}}\left(
 -c\mathbf{1}_{\{X_t-\dl{X}_t=0\}}\mathbf{1}_{\{c<0\}}+\int_{0}^\infty (1-e^{-u})
\frac{(\mu_-)^2}{\nu_-}e^{-\frac{\mu_-}{\nu_-}(u+X_t -\dl{X}_t)}(u+X_t -\dl{X}_t)^{-1}du
\right).
\end{eqnarray}
}
\end{emp}

\comment{
\begin{figure}
\begin{center}
\includegraphics[scale=0.7]{./Fig/VG_3_t.png}
\caption{Sample path of variance gamma process.}
\end{center}
\end{figure}
}

We next provide an application
of Theorem \ref{thm: main} in credit risk modelling. The credit spread $S(t,h)$ of a defaultable name over the time interval $[t,
t+h]$ is defined by
$$ S(t,h) := -\frac{1}{h}\ln\left(1-\Prob(t<\tau \leq t+h|\mathcal{G}_t)\right),
$$
where $\Prob\left(t<\tau \leq t+h|\mathcal{G}_t\right)$ is the conditional
default
probability  given $\tau>t$.
Using the Taylor expansion, we can find the instantaneous credit spread  $s(t)$
as
\begin{eqnarray*}
s(t):= \lim_{h \downarrow 0}S(t,h) =\lim_{h \downarrow 0}\frac{1}{h}\Prob\left(t<\tau
\leq t+h|\mathcal{G}_t)\right) =\td{\lambda}_t.
\end{eqnarray*}
Theorem~\ref{thm: main} says that
 $s(t)$ is positive and finite almost surely and is given by
$$s(t) =-c\mathbf{1}_{\{X_t-\dl{X}_t=0\}}\mathbf{1}_{\{c<0\}}+\Pi(X_t-\dl{X}_t),
$$ 
which conforms to the market observation that the instantaneous credit spread
remains positive and finite even though the bond is near its maturity and that the bond price often drops around the time of default due to uncertainties about the closeness of the current asset value to the default threshold. For more details of the instantaneous credit spread and its term structure, see \cite{DuffieLando03,Giesecke06}.

We next give a numerical example to illustrate the results.  We take the variance gamma process $VG(c,\nu,\sigma,\theta)$  in Example~\ref{emp: VG}.
The data used are $(c,\nu,\sigma,\theta) = (-0.02,0.1,0.15,0.01)$. 
Figure~\ref{fig: X_intensity} displays for $t\in [0,5]$ a sample path of the asset return process $X$, the running minimal process $\dl{X}$ and the resulting intensity process
$\lambda$. Figure~\ref{fig: X_intensity} also shows  the distance $X_t - \dl{X}_t$ and its contribution $\Pi(X_t - \dl{X}_t)$ to  the intensity. We can observe the reciprocal relation of the intensity  $\lambda_t$ and the distance  $X_t - \dl{X}_t$, which is consistent with the observation in the credit market. Note that $\Pi(\cdot)$ on $\mathbb{R}_+$ is bounded above by $\Pi(0)$ that is fully determined by the L\'evy measure of $X$.  The upper bound $\Pi(0)$ is reached when $X_t - \dl{X}_t = 0$, i.e. the process $X$ reaches a new minimal level, and the  intensity  $\lambda_t$ at that time is above $\Pi(0)$ by the amount $|c|$ as the drift parameter $c<0$.  
Figure~\ref{fig: creditSpread_termStructure}, using the same sample path of Figure~\ref{fig: X_intensity},  shows the term structure of the credit spread $h\mapsto S(t,h)$ at time $t=0.5$, starting from $S(t,0)=\lambda_t$.

\begin{figure}[h!]
\centering
\includegraphics[trim = 10mm 45mm 5mm 20mm, clip, width=0.9\textwidth]{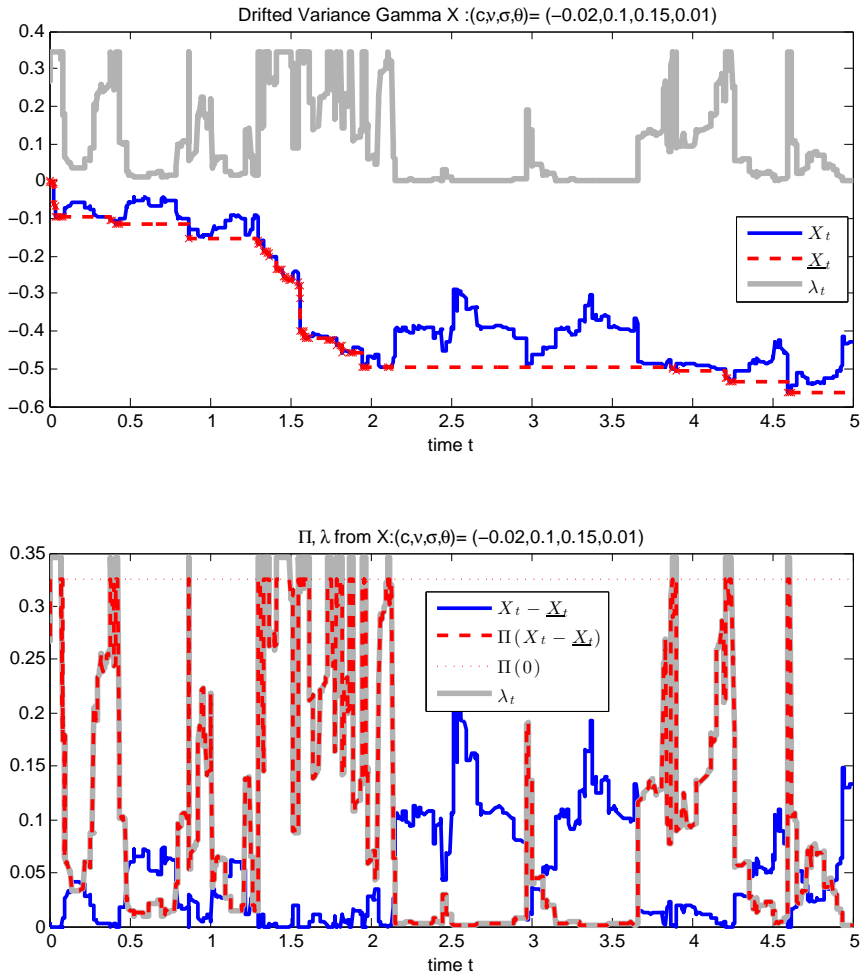}
\caption{The asset return process $X$ as in Example~\ref{emp: VG},  the distance process $X - \dl{X}$ and  the intensity process $\lambda$. The data used are $(c,\nu,\sigma,\theta)=(-0.02,0.1,0.15,0.01)$. }
\label{fig: X_intensity}
\end{figure}

\begin{figure}[h!]
\centering
\includegraphics[trim = 10mm 70mm 5mm 60mm, clip, width=0.6\textwidth]{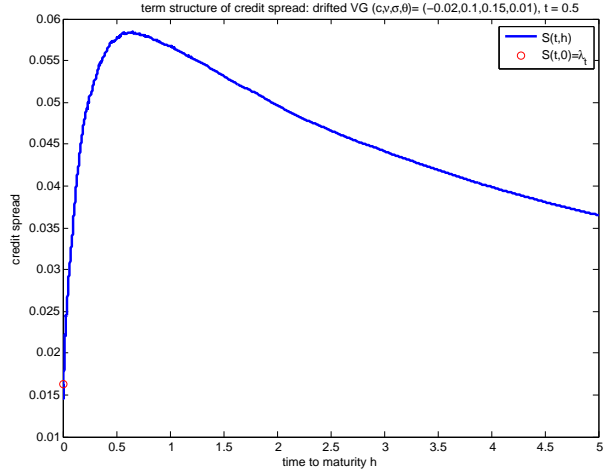}
\caption{The term structure of credit spread $S(t,h)$: the asset return process  $X$ as in Example~\ref{emp: VG} with the data 
$(c,\nu,\sigma,\theta)=(-0.02,0.1,0.15,0.01)$, at $t=0.5$,
 $X_t - \dl{X}_t =
0.0585$.}\label{fig: creditSpread_termStructure}
\end{figure}

\section{Proof of Theorem~\ref{thm: main}}\label{sec: proof main theorem}
Theorem~\ref{thm: main} is proved in four steps, detailed in the following
subsections.
Subsection 3.1 shows the relation of the likelihood processes under different
filtration (Lemma \ref{k^n}), Subsections 3.2 and 3.3  establish the existence
of the limit process for a spectrally negative L\'evy process with finite
variation (Proposition \ref{prop:td_kn_Convg})  and for a general L\'evy
process with finite variation
(Proposition \ref{prop:Squeezing_positiveDrift}), and Subsection 3.4 confirms
the indistinguishability of the instantaneous likelihood process and the intensity process
using Aven's condition.

\subsection{Compensators and Likelihood Processes under Different Filtrations}
The conditional survival probability at each time $t$ is given by
\begin{eqnarray*}\label{eqn: defLt}
Z_t  := \mathbb{P}(\tau >t | \mathcal{F}_t) = \mathbb{P}(\dl{X}_t > D | \mathcal{F}_t)
= e^{\dl{X}_t}.
\end{eqnarray*} 
It is known (\cite[Chapter VI, Theorem 11]{Protter05}) that there exists a unique, increasing, $\G$-predictable process $A$, the
$\G$-compensator of $N$, such that the difference of $A$ and $N$ is a uniformly
integrable $\G$-martingale.
Our objective is to find $A$. 

Let $Z_{t-} = \lim_{s\uparrow t} Z_s$ and $Z_{0-}=1$. Define a nondecreasing
$\mathbb{F}$-predictable process $A$ by
\begin{eqnarray*}\label{eqn:trend}
A_t = \int_0^t \frac{dK_s}{Z_{s-}},
\end{eqnarray*}
where $K_t$ is the unique, increasing, $\F$-predictable compensator of $\F$-submartingale
$1-Z_t = \Prob(\tau \leq t|\Ft_t)$. 
\comment{
From this definition, we observe that the regularity of $A$ is consistent
of $K$. $A$ is continuous/absolutely continuous if and only if $K$ is continuous/absolutely
continuous. In the case where $Z$ is continuous, then $K=L$ as $Z$ is $\F$-predictable
thus $A_t = -\ln Z_t$ continuous.
}

\begin{thm}[\cite{JY78}] \label{thm3.1}
The process $N-A^\tau$ is a $\G$-martingale, where $A^\tau = (A_{t\wedge
\tau})_{t\geq 0}$. \label{prop:JY}
\end{thm}
Theorem \ref{thm3.1} shows that one can transform the problem of finding
the $\G$-compensator of $N$ into the problem of finding the $\F$-compensator
of $Z$. 
If $Z_t$ is a continuous process, then $K_t = -Z_t$ and $A_t = -\ln(Z_t)$.
If $Z$ is discontinuous, then finding $K$ is nontrivial, see \cite{GZ08}.

The next result characterizes the likelihood processes under
different filtrations.

\begin{lemma} \label{k^n}
For any  L\'evy process $X$, $h>0$, denote 
\begin{eqnarray*}
k^h_t &:=& \frac{1}{h}\E\left[ K_{t+h}-K_t  |\Ft_t \right]
\quad\mbox{and}\quad
\lambda^h_t :=\frac{1}{h}\E\left[N_{t+h}-N_t |\mathcal{G}_t \right].
\end{eqnarray*}
Then, 
\begin{eqnarray}\label{k^h}
k^h_t = e^{\dl{X}_t} \frac{1}{h}  \mathbb{E}\left[1-e^{-\left(y -\dl{X}_{h}
\right)^+} \right]\Big|_{y=\dl{X}_{t}-X_t}
\quad\mbox{and}\quad
\lambda^h_t = 1_{\{\tau>t\}}e^{-\dl{X}_t} k^h_t.
\end{eqnarray}
\end{lemma}

\begin{proof}
Since
\begin{eqnarray*}
\dl{X}_{t+h}-\dl{X}_{t} &=& \inf_{u\in[t,t+h]}X_u \wedge \dl{X}_t -\dl{X}_{t}\\
&=& \inf_{u\in[0,h]}(X_{t+u}-X_t) \wedge (\dl{X}_t-X_t) -(\dl{X}_{t}-X_t)\\
&=& -\left( (\dl{X}_{t}-X_t) -\inf_{u\in[0,h]}(X_{t+u}-X_t) \right)^+,
\end{eqnarray*}
we have
\begin{eqnarray*}
\mathbb{E}\left[e^{\dl{X}_{t+h}}-e^{\dl{X}_{t}}|\mathcal{F}_t \right] &=&
e^{\dl{X}_t}\mathbb{E}\left[e^{\dl{X}_{t+h}-\dl{X}_{t}}-1|\mathcal{F}_t \right]\\
&=& e^{\dl{X}_t}\mathbb{E}\left[e^{-\left( (\dl{X}_{t}-X_t) -\inf_{u\in[0,h]}(X_{t+u}-X_t)
\right)^+}-1|\mathcal{F}_t \right]\\
&=& e^{\dl{X}_t}\mathbb{E}\left[e^{-\left(y -\dl{X}_h \right)^+}-1 \right]\Big|_{y=\dl{X}_{t}-X_t},
\end{eqnarray*}
where the last equality comes from the independent and stationary increment
property of L\'evy process $X$ and adaptedness of $X$ and $\dl{X }$ in $\F$.
Since  $K$ is the $\F$-compensator of $1- Z$, the Doob-Meyer decomposition
 says that 
\begin{eqnarray*}
\mathbb{E}\left[K_{t+h}-K_t|\mathcal{F}_t \right] = -\mathbb{E}\left[Z_{t+h}-Z_t|\mathcal{F}_t
\right] = -\mathbb{E}\left[e^{\dl{X}_{t+h}}-e^{\dl{X}_{t}}|\mathcal{F}_t
\right]. 
\end{eqnarray*}
Combining the above  gives  $k^h_t$ in (\ref{k^h}).

Next, by the optional projection theorem (Theorem 14, Chap.VI, \cite{Protter05}
and  \cite{Giesecke06}), we know that if a random variable $\xi$ is nonnegative
and integrable, then for each $t\geq 0$, the right continuous version of
$\E[\xi|\Gt_t]$ is given by
\begin{eqnarray}\label{eqn: optional_projection}
\E[\xi|\Gt_t] = 1_{\{\tau>t\}}\frac{1}{Z_t} \E\left[\xi 1_{\{\tau>t\}}|\Ft_t\right]+\xi
1_{\{\tau \leq t\}} \quad a.s.
\end{eqnarray}

Therefore, using the tower property of the expectation and the fact that
 $K$ is a $\F$-compensator of $1-Z$, we have
\begin{eqnarray*}
\lambda^h_t
&=&  \frac{1}{h}\E\left[N_{t+h}-N_t |\mathcal{G}_t \right]\\
&=&  1_{\{\tau>t\}}\frac{1}{h}{1\over Z_t}\E\left[1_{\{t<\tau\leq t+h\}}|\mathcal{F}_t
\right]\\
&=&    1_{\{\tau>t\}}\frac{1}{h}\frac{1}{Z_{t}} \E\left[ Z_t-Z_{t+h}  |\Ft_t
\right] \\
&=&1_{\{\tau>t\}}\frac{1}{Z_{t}}  \frac{1}{h}\E\left[ K_{t+h}-K_t  |\Ft_t
\right]. \\
&=& 1_{\{\tau>t\}}e^{-\dl{X}_t} k^h_t.
\end{eqnarray*}
This gives  $\lambda^h_t$ in (\ref{k^h}).\end{proof}

The next result follows immediately from Lemma \ref{k^n}.
\begin{cor}\label{cor: instantaneousProcess} Assume that 
$\td{k}_t:=\lim_{h\downarrow 0} k^h_t $ exists  for all $t$ a.s., then the
instantaneous likelihood process on $\{\tau>t\}$ is given by

 $$\td{\lambda}_t := \lim_{h\downarrow 0}\frac{1}{h}\Prob(t< \tau \leq t+h|\mathcal{G}_t)=
e^{-\dl{X}_t}\td{k}_t.$$
\end{cor}

\begin{rmk}
Note that $\xi 1_{\{\tau \leq t\}}$ in (\ref{eqn: optional_projection}) is
$\Gt_t$ measurable by definition. Indeed, since $\xi$ and $\tau$ are random
variables on $(\Omega,\mathcal{G},\Prob)$,
then $\xi 1_{\{\tau \leq t\}}$ is $\mathcal{G}$-measurable. To show it is
$\mathcal{G}_t$ measurable, it is equivalent to show $\forall b\in \mathbb{R},
\quad B(b):= \{\omega: \xi(\omega)1_{\{\tau \leq
t\}}(\omega) \leq b\} \in \Gt_t$.
Note that
\begin{eqnarray*}
B(b) \cap  \{t<\tau\}&=& \{\xi 1_{\{\tau \leq t\}} \leq b\}\cap  \{t < \tau
\}\\ 
&=& \left( \{\xi \leq b, \tau \leq t\} \cup  \{ 0 \leq b, t < \tau \} \right)\cap
\{t < \tau \}\\
&=&  \{ 0 \leq b, t < \tau \} \\
&=& 1_{\{b<0\}} \emptyset + 1_{\{b\geq 0\}} \{ t < \tau \} \\
&=&  1_{\{b<0\}} \left(\emptyset \cap  \{ t < \tau \}\right)  + 1_{\{b\geq
0\}} \left(\Omega \cap \{ t < \tau \} \right).
\end{eqnarray*}
Since $\emptyset, \Omega \in \Ft_t$ for all $t\geq 0$, we can take
$B_t(b) := 1_{\{b<0\}} \emptyset + 1_{\{b\geq 0\}} \Omega \in \Ft_t$,
such that 
\begin{eqnarray*}
B(b) \cap \{\tau > t\} = B_t(b) \cap  \{\tau > t\}.
\end{eqnarray*}
Therefore, we have $B(b) \in \Gt_t$.
\end{rmk}

\subsection{Spectrally Negative L\'evy Process with Finite Variation}\label{subsec:
SNFVLP}

Let $X$ be a spectrally negative L\'evy process with finite variation, then $X$ has a representation 
  \cite[page 56]{Kyprianou06}
\begin{eqnarray}\label{eqn:dsub}
X_t = ct - S_t,
\end{eqnarray}
 where $c>0$  and $S$ is a pure jump subordinator with L\'evy measure $\pi$.  (\ref{eqn:dsub}) is a special case of (\ref{eqn:
decompFVLP}) with $\pi'=0$ and $c>0$.  The L\'evy measure of $X$ is $\pi_X(dx)
= \pi(d(-x))= \pi((-x,-x+dx])$ on $\mathbb{R}_-$ and if $\pi$ admits a density
$\nu$ then $\pi(-dx) = \nu(-x)dx$.
The following concept is needed in analysing the path property of $\dl{X}$.

\begin{definition}[\cite
{Kyprianou06}]  Let $X$ be a L\'evy process.
A point $x\in \mathbb{R}$ is said to be {\it irregular} for  an open or closed
set $B$ if
$\Prob_x\left(\tau^B=0\right)=0$,
where the stopping time $\tau^B = \inf\{t>0: X_t \in B\}$.
\end{definition}

We know (\cite[Chapter 9, Proposition 15]{Doney01}) that for $X$
defined in
 (\ref{eqn:dsub}), $0$ is irregular for $(-\infty,0)$. Hence, starting at
$0$, it takes $X$ strictly positive
time to reach $(-\infty,0)$.
If we define $T_1 := \inf\{t>0: X_t <0\}$, then $\Prob(T_1 > 0) = 1$.
$T_1$ is the first jump time of $\dl{X}$ but may not be the first jump time
of $X$. 
We observe that  $\dl{X}$ is a pure-jump process as  $\dl{X}$ can only move
when $S$ jumps and  $\dl{X}$ cannot jump to a pre-specified level on $(-\infty,0)$
as $X$ can not, see  \cite[Exercise 5.9]{Kyprianou06}. Hence, the jump size
of $\dl{X}$ has no atoms and is strictly negative. The number of jumps of
 $\dl{X}$ on  the interval $[0,t]$, i.e., 
$n_t := \#\{s\in(0,t]: X_s = \dl{X}_s\}$, 
 is a discrete set and is a.s. finite. Moreover, we denote the arrival times
of $n_t$ by $(T_i)_i$,  the inter-arrival times by $(\delta_i)_i$, and the
jump sizes by $(\xi_i)_i$. Then we have the following lemma.

\comment{
This representation of the minimum will be useful for what follows.

We have just seen that $\dl{X}$ is a pure-jump process. In this section,
we shall show that it is also a renewal reward process.

\begin{definition}[Renewal Reward Process]
Let $\delta_1, \delta_2,\ldots,$ be a sequence of positive i.i.d.\ random
variables. Denote $T_n = \sum_{i=1}^n \delta_i$ for each $n\in \mathbb{N}$,
then $n_t := \sum_{i=1}^\infty \1_{\{T_i\leq t\}}$ is called a renewal process.

Let $Y_1,Y_2,\ldots$ be a sequence of i.i.d. random variables, then 
\begin{eqnarray*}
Z_t &:=& \sum_{i = 1}^{n_t}Y_i
\end{eqnarray*}
is a renewal-reward process. 

\end{definition}

\begin{rmk}
$(\delta_i)_i$ and $(Y_i)_i$ do not need to be independent. In particular,
$Y_i$ can be a function of $\delta_i$.
\end{rmk}
}

\begin{lemma}
For $X$ defined in  (\ref{eqn:dsub}), $\dl{X}$ can be written as a renewal-reward
process
\begin{eqnarray*}\label{eqn:rrp}
\dl{X}_t = -\sum_{i=1}^{n_t}\xi_i,
\end{eqnarray*}
where $(\delta_i, \xi_i)$ are i.i.d. random variables.
\end{lemma}

\begin{proof}
The analysis above shows that  $\dl{X}$ is a non-explosive marked point process
and can be written as  $\dl{X}_t = -\sum_{i=1}^{n_t}\xi_i$, where $-\xi_i
= \Delta \dl{X}_{T_i} =  \dl{X}_{T_i} - \dl{X}_{T_{i-1}} =X_{T_i} - X_{T_{i-1}}$.
Since $(T_i)_i$ are also jump times of L\'evy process $X$ and are stopping-times.
We have that $(\delta_i,\xi_i)_i$ are i.i.d. random variables due to the strong
Markov property of $X$.
\end{proof}

Instead of investigating the exact law of $\dl{X}$, we only need to analyse
the small-time behaviour of the process, which can be done with the help
of the next result, called
the Ballot Theorem \cite[Proposition
2.7]{Bertoin00}.

\comment{
\begin{thm}[Ballot Theorem, proposition 2.7, Bertoin\cite{Bertoin00}] We
have for every $t>0$, $y\leq 0$, and $z>-y$,
\begin{eqnarray*}
\Prob\left(l_1\in dt, Y_{l_1-}\in dy, \Delta Y_{l_1} \in dz\right)= \frac{-y}{ct}\Prob(Y_t
\in dy)\pi(dz)dt
\end{eqnarray*}
where $\Delta_t = Y_t - Y_{t-}$, $\pi$ is the L\'evy measure of $S$ and $l_1
= \inf\{t>0: Y_t > 0\}$.
\end{thm}

We have the following corollary by symmetry.}

\begin{lemma}[\cite{Bertoin00}]
Let $X$ be defined in (\ref{eqn:dsub}) and $T_1 := \inf\{t>0: X_t <0\}$.
Then, for every $t>0$, $z\geq 0$, and $u<-z$,
\begin{eqnarray*}
\Prob\left(T_1 \in dt, X_{T_1-}\in dz,\Delta X_{T_1}\in du \right) = \frac{z}{ct}\Prob(X_t
\in dz)\pi(-du)dt,
\end{eqnarray*} 
where $\Delta X_t = X_{t}-X_{t-}$ and $\pi$ is the L\'evy measure of $S$.
\end{lemma}

Hence the joint distribution of $(T_1,X_{T_1})$ is given by 
\begin{eqnarray}\label{eqn:BallJoint2}
\mathbb{P}(T_1\in dt,X_{T_1}\in dw)
&=& \left(\int_{z\in(0,\infty)} z\pi(z+d(-w))\mathbb{P}(X_t\in dz)\right)\frac{1}{ct}dt
\end{eqnarray}
for $w \leq 0$.
The following is another version of the Ballot theorem:
\begin{eqnarray*}
\Prob(T_1>t,X_t \in dx) = \frac{x}{ct}\Prob(X_t \in dx)
\end{eqnarray*} 
for every $t>0$ and $x \in [0,\infty)$. Since $X_t = ct-S_t \leq ct$, we
have
\begin{eqnarray*}
\Prob(T_1 >t) = \frac{1}{ct}\int_0^\infty x\Prob(X_t \in dx) = \frac{1}{ct}\int_0^{ct}
x\Prob(X_t \in dx) 
= \frac{1}{c}\E\left[ \mathbf{1}_{\{0\leq X_t \leq ct\}} \frac{X_t}{t} \right].
\end{eqnarray*}
Note as $\lim_{t \downarrow 0} \frac{S_t}{t} = 0$ a.s., we have for almost
all $\omega$, there exists $t_0(\omega)$, such that for all $t \in [0,t_0(\omega)]$,
$S_t(\omega) \leq c t$, hence $0\leq X_t(\omega)=ct-S_t(\omega) \leq ct$
and 
\begin{eqnarray}\label{eqn:limit_indicator}
 \lim_{t\downarrow 0}\mathbf{1}_{\{0\leq X_t \leq ct \}} = 1 \quad a.s.
\end{eqnarray} 
The dominated convergence theorem leads to
\begin{eqnarray}\label{eqn:limit_T1}
\lim_{t \downarrow 0} \Prob(T_1 >t) = \frac{1}{c}\E\left[ \lim_{t \downarrow
0}  \mathbf{1}_{\{0\leq X_t \leq ct\}} \frac{X_t}{t} \right] = \frac{1}{c}\cdot
c = 1.
\end{eqnarray}

\begin{proposition}\label{prop:td_kn_Convg}
Let $X$ be defined in (\ref{eqn:dsub}) and  let Assumption~\ref{asm: negLevyMeasure}
be satisfied.  Then the following limit exists for all $t$ a.s.
$$
\tilde{k}_t:=\lim_{h\downarrow 0} k^h_t = e^{\dl{X}_t}\Pi(X_t - \dl{X}_t),
$$
where $k_t^h$ is defined in (\ref{k^h}) for $h>0$ and $\Pi$ is defined in
(\ref{Pi}). 
\end{proposition}


\begin{proof}
Recall that  $\dl{X}_t = -\sum_{i=1}^{n_t}\xi_i$ is a renewal-reward process,
where jump size $\xi_i$ and inter-arrival times $\delta_i$ are positive random
variables for all $i$, and $(\delta_i, \xi_i)$ are i.i.d. random variables.

By (\ref{eqn:limit_T1}), denote by $F$ the distribution function of $\delta_i$.
Then
\begin{eqnarray*}
\lim_{t\downarrow 0} F(t) =\lim_{t\downarrow 0} P(T_1\leq t) = 0 =F(0).
\end{eqnarray*}
Hence, $F$ is right continuous at zero, i.e.  $F(0)= F(0+) = 0$.

Denote by, for $t>0$ and $y\leq 0$,
\begin{eqnarray*}
\Lambda^0_t(y) &:=& \frac{1}{t}\mathbb{E}\left[1 - e^{-(y-\dl{X}_{t})^+}\right]
\\
\Lambda_t^1(y) &:=& \frac{1}{t}\mathbb{E}\left[1_{\{n_t=1\}}(1 -e^{-(\xi_1+y)^+})\right]
\\
\Lambda_t^2(y) &:=& \frac{1}{t} \mathbb{E}\left[\sum_{k=2}^\infty 1_{\{n_t
= k\}}\left(1 - e^{-\left(\sum_{i=1}^k\xi_i +y\right)^+}\right)\right].
\end{eqnarray*}
We have
\begin{eqnarray*}
\Lambda^0_t(y) = \mathbb{E}\left[1 - e^{-\left(y+\sum_{i=1}^{n_t}\xi_i\right)
^+}\right] = \Lambda^1_t(y) + \Lambda^2_t(y) .
\end{eqnarray*}

We next show that for $y\leq 0$,
\begin{eqnarray}
&&\lim_{t\rightarrow 0}\Lambda^1_t(y)=\Pi(-y), \mbox{ and  } \lim_{t\rightarrow
0}\Lambda^2_t(y) = 0. \label{eqn: nt_1}
\end{eqnarray}
Then, (\ref{eqn: nt_1}) gives the required conclusion.

Since $\delta_1$ and $\delta_2$ are independent, also noting (\ref{eqn:BallJoint2}),
we have 
\begin{eqnarray}
&& \mathbb{E}\left[1_{\{n_t=1\}}(1 - e^{-(\xi_1+y)^+})\right] \label{eqn:nt1_base}\\
&=& \mathbb{E}\left[1_{\{T_1 \leq t\}}1_{\{T_2-T_1 > t-T_1\}}(1 - e^{-(\xi_1+y)^+})\right]
\nonumber\\
&=&\int_0^t \int_{-y}^\infty \bar{F}_{T_1}(t-s)(1-e^{-x-y}) \mathbb{P}(T_1\in
ds,X_{T_1}\in d(-x))\nonumber\\
&=& \int_{s=0}^t \int_{x=-y}^\infty \bar{F}_{T_1}(t-s)(1-e^{-x-y})\left(\int_{z=0}^\infty
z \pi(z+dx))\mathbb{P}(X_s\in dz)\right) \frac{1}{cs} ds \nonumber\\
&=&\frac{1}{c}\int_{s=0}^t \bar{F}_{T_1}(t-s)\int_{z=0}^\infty \left(\int_{u=0}^\infty
(1-e^{-u})\pi(z-y+du)\right) z\frac{\mathbb{P}(X_s \in dz)}{s}ds\nonumber\\
&=& \frac{1}{c}\int_{s=0}^t \bar{F}_{T_1}(t-s) \left(\int_{z=0}^{cs} \Pi(z-y)
z\frac{\mathbb{P}(X_s \in dz)}{s}\right)ds.\nonumber
\end{eqnarray}
The last equality is due to $X_s = cs-S_s \leq cs$.

Since $S$ is a pure jump subordinator, we have (\cite[Lemma 4.11]{Kyprianou06}) $\lim_{t\rightarrow 0}\frac{S_t}{t} = 0 \quad \mbox{a.s.}$, 
which implies 
\begin{eqnarray}\label{eqn:sub_0}
\lim_{t\rightarrow 0}\frac{X_t}{t} = c. 
\end{eqnarray}
Using (\ref{eqn:sub_0}) and (\ref{eqn:limit_indicator}), the dominated convergence
theorem, continuity of $\Pi(\cdot)$, and $X_{0+} = 0$, we obtain
\begin{eqnarray*}
\lim_{s\downarrow 0}\int_{z=0}^{cs}  z \Pi(z-y)\frac{\mathbb{P}(X_s \in dz)}{s}
&=& \lim_{s\downarrow 0} \mathbb{E}\left[1_{\{0\leq X_s \leq cs\}}\frac{X_s}{s}\Pi(X_s-y)\right]
\\
&=& \mathbb{E}\left[\lim_{s\downarrow 0} 1_{\{0\leq X_s \leq cs\}}\frac{X_s}{s}\Pi(X_s-y)\right]\\
&=& c\Pi(-y).
\end{eqnarray*}

Taking the limit in (\ref{eqn:nt1_base}) gives
\begin{eqnarray*}
\lim_{t\downarrow 0} \Lambda_t^1(y) = \frac{1}{c}\lim_{s\downarrow 0}\int_{z=0}^{cs}
z \Pi(z-y)\frac{\mathbb{P}(X_s \in dz)}{s} = \Pi(-y). \\
\end{eqnarray*}
Here we have used the fact that if $g(\cdot)$ is a nonnegative function and
$\bar{F}(0+)=1$, then
\begin{eqnarray*}
\frac{1}{t}\int_0^t \ul{F}(t) g(s)ds\leq \frac{1}{t}\int_0^t \ul{F}(t-s)
g(s)ds \leq \frac{1}{t}\int_0^t g(s)ds
\end{eqnarray*}
and\begin{eqnarray*}
\lim_{t\rightarrow 0}\frac{1}{t}\int_0^t \ul{F}(t-s) g(s)ds = \lim_{t\rightarrow
0}\frac{1}{t}\int_0^t g(s)ds = g(0+).
\end{eqnarray*}

We have proved the first  limit in (\ref{eqn: nt_1}). We next prove the second
limit in (\ref{eqn: nt_1}). Since 
\begin{eqnarray*}
\Lambda^0_t(0) = \frac{1}{t}\E\left[1 - e^{\dl{X}_t}\right] \leq  \frac{1}{t}\E\left[1
- e^{-S_t}\right] = \frac{1}{t}\left(1 - e^{-\Pi(0)t}\right) 
\leq  \Pi(0)
\end{eqnarray*}
and 
\begin{eqnarray*}
0\leq\Lambda^1_t(0) \leq \Lambda^0_t(0)\leq \Pi(0),
\end{eqnarray*}
the first limit in (\ref{eqn: nt_1}) implies
\begin{eqnarray*}
\lim_{t\rightarrow 0} \Lambda^0_t(0) = \lim_{t\rightarrow 0} \Lambda^1_t(0)=
\Pi(0),
\end{eqnarray*}
therefore
\begin{eqnarray*}\label{eqn:Lambda2y0}
\lim_{t\rightarrow 0} \Lambda^2_t(0) = 0.
\end{eqnarray*}

On the other hand, we know 
\begin{eqnarray*}
0 \leq \Lambda^2_t(y) \leq \Lambda^2_t(0) \quad \mbox{for all $y\leq 0$},
\end{eqnarray*}
which proves the second limit in (\ref{eqn: nt_1}). 
Hence, $\Lambda^0_t(y) = \Pi(-y)$ and $\td{\lambda}_t = \Lambda^0_t(X_t -
\dl{X}_t)=  \Pi(X_t - \dl{X}_t)$. 
\end{proof}

\begin{rmk}{\rm
Note that  $\Pi(\cdot)$ is continuous as $\pi(dx)$ is. $\Pi(0) = -\ln \E[e^{-S_1}]
$ is the Laplace exponent of $S$ from the L\'evy-Khintchine formula, and
$0<\Pi(x) \leq \Pi(0) < \int_0^\infty u \pi(du)< \infty$ for all $x\geq 0$
by Assumption \ref{asm: negLevyMeasure}. Therefore, $\Pi$ is bounded on $\mathbb{R}_+.$
}
\end{rmk}


\comment{
\begin{rmk}[$X$ and $\dl{X}$]
{\rm Note that $\dl{X}$ is a simple point process and  that, in a small time
interval, if there is a jump for $\dl{X}$, this jump may be not the first
jump of $X$. }
\comment{
However, for t small before the first jump of $\dl{X}$, we still have $(y-\dl{X}_t)^+
= (y-{X}_t)^+$, thus
\begin{eqnarray*}
\frac{1}{t}\mathbb{E}[e^{-(y-\dl{X}_t)^+}-1] = \frac{1}{t}\mathbb{E}[e^{-(y-{X}_t)^+}-1]=
\int_{-\infty}^y (e^{-(y-x)}-1)\frac{\Prob(X_t \in dx)}{t} = 
\end{eqnarray*}
}
\end{rmk}
}

\subsection{L\'evy Process with Finite Variation}\label{subsec: FVLP}
Suppose $X$ is a L\'evy process with finite variation. It then has a representation
(\ref{eqn: decompFVLP}) and we assume that  Assumption~\ref{asm: negLevyMeasure}
holds. 
Note that the path properties and techniques used in Subsection 3.2 no longer
hold.
In (\ref{eqn: decompFVLP}), denote the drift and negative jump components
as 
\begin{eqnarray*}\label{eqn: Ztc}
Z_t(c) := ct - S_t.
\end{eqnarray*}
Then we first claim the following result for $Z_t(c)$.
\comment{
Recall in spectrally negative case, we have the existence of limit of $\phi_{t}$
as $t\rr 0$ when $c<0$. In this part, we first prove for the case where $c\geq
0$. In this way we generalize result for the drift coefficient $c$ from negative
to any real value. }
\begin{lemma}\label{prop:Zyleq0}
For any $c \in \mathbb{R}$ and $y\leq 0$ the following limit exists:
\begin{eqnarray}\label{eqn: Zmin}
\lim_{h\rr 0}\frac{1}{h}\E\left[1 - e^{-(y-\dl{Z}_h(c))^+}\right] &=&  -c\mathbf{1}_{\{y=0\}}\mathbf{1}_{\{c<0\}}
+ \Pi(-y),
\end{eqnarray}
where $\Pi$ is defined in (\ref{Pi}).
\end{lemma}

\begin{proof} For $c>0$ the limit (\ref{eqn: Zmin}) has been proved in the
previous subsection. We now consider the case of $c\leq 0$. 
Note that $Z_h(c)$ is decreasing in $h$ and  $\dl{Z}_h(c)=Z_h(c)$. We split
the proof into two cases.

(i) $y=0$: We have
\begin{eqnarray*}
\lim_{h\rr 0}\frac{1}{h}\E\left[1 - e^{-(y-\dl{Z}_h(c))^+}\right] =\lim_{h\rr
0}\frac{1}{h}\E\left[1 - e^{ch-S_h}\right]  &=& -c + \Pi(0).
\end{eqnarray*}

(ii) $y<0$: Take the function $f(x) := 1 - e^{-(y+x)^+}$ which is bounded,
continuous, and vanishes in a neighbourhood of zero: take $\epsilon < -y$,
then for any $x\in (0,\epsilon)$, we have $f(x) = 0$.
Hence
\begin{eqnarray*}
\lim_{h\rr 0}\frac{1}{h}\E\left[ 1- e^{-(y-\dl{Z}_h(c))^+}\right]&=& \lim_{h
\rr 0} \frac{1}{h}\E\left[f(-Z_h(c))\right] 
= \int_{\mathbb{R}} f(x) \pi(dx).
\end{eqnarray*}
The second equality is due to \cite[Corollary 8.9]{Sato99} and $\pi$ being
the L\'evy measure of $-Z_h(c)$. Therefore,
\begin{eqnarray*}
\int_{\mathbb{R}} f(x) \pi(dx)  &=& \int_{-y}^\infty (1 - e^{-(y+x)})\pi(dx)=
\int_0^\infty (1-e^{-u})\pi(-y+du) = \Pi(-y),
\end{eqnarray*}
which proves (\ref{eqn: Zmin}).
\end{proof}

\begin{proposition}\label{prop:Squeezing_positiveDrift}
Let $X_t$ be defined in (\ref{eqn: decompFVLP})
and  let Assumption~\ref{asm: negLevyMeasure} be satisfied. Then  
the following limit exists for all $t$ a.s.
\begin{equation}\label{eqn: td_k_temp}
\tilde{k}_t:=\lim_{h\downarrow 0} k^h_t = e^{\dl{X}_t}\left(-c\mathbf{1}_{\{X_t-\dl{X}_t=0\}}\mathbf{1}_{\{c<0\}}+\Pi(X_t-\dl{X}_t)
\right), 
\end{equation}
where $k_t^h$ is defined in (\ref{k^h}) for $h>0$ and $\Pi$ is defined in
(\ref{Pi}). The instantaneous likelihood process $\tilde{\lambda}_t$ defined in Corollary \ref{cor: instantaneousProcess} is given by
\begin{equation}\label{eqn: td_lambda_temp}
\tilde{\lambda}_t = -c\mathbf{1}_{\{X_t-\dl{X}_t=0\}}\mathbf{1}_{\{c<0\}}+\Pi(X_t-\dl{X}_t). 
\end{equation}

\end{proposition}

\begin{proof}
The expression of $\tilde{\lambda}_t$ in (\ref{eqn: td_lambda_temp}) is an immediate result of Corollary \ref{cor: instantaneousProcess} and (\ref{eqn: td_k_temp}). To prove (\ref{eqn: td_k_temp})  we only need to show that for all $y\leq 0$,
\begin{eqnarray}\label{temp1}
 \lim_{h\rightarrow 0 }\frac{1}{h}\E\left[1 - e^{-(y-\dl{X}_h)^+}\right]
=-c\mathbf{1}_{\{y=0\}}\mathbf{1}_{\{c <  0\}}+\Pi(-y). \label{eqn_t1}
\end{eqnarray}

Since  $f(x)= 1 - e^{-(y-x)^+}$ is a decreasing function of $x$ on $\mathbb{R}_-$
and $X_h = ch-S_h + S'_h \geq ch-S_h = Z_h(c)$ for all $\omega$ and $h> 0$,
 we have $\dl{X}_h \geq \dl{Z}_h(c)$ and 
\begin{eqnarray*}
\frac{1}{h}\E\left[1 - e^{-(y-\dl{X}_h)^+}\right]&\leq & \frac{1}{h}\E\left[1
- e^{-(y-\dl{Z}_h(c))^+}\right].
\end{eqnarray*}
Using Lemma~\ref{prop:Zyleq0} we obtain
\begin{eqnarray*}
\limsup_{h\rr 0}\frac{1}{h}\E\left[1 - e^{-(y-\dl{X}_h)^+}\right] \leq -c\mathbf{1}_{\{y=0\}}\mathbf{1}_{\{c
<  0\}}+\Pi(-y).
\end{eqnarray*}

Take any $\epsilon>0$, on the set $\{S'_h \leq \epsilon h \}$ we have:
\begin{eqnarray*}
X_h = ch-S_h + S'_h\leq ch - S_h + \epsilon h = Z_h(c+\epsilon),
\end{eqnarray*}
which yields
\begin{eqnarray*}
\dl{X}_h \leq \dl{Z}_h(c+\epsilon) \quad \mbox{on $\{S'_h \leq \epsilon h
\}$}.
\end{eqnarray*}
Moreover, as $\dl{X}_h \leq 0$ for all $h\geq 0$, we have almost surely,
\begin{eqnarray*}
\dl{X}_h = 1_{\left\{S'_h\leq \epsilon h\right\}} \dl{X}_h + 1_{\left\{S'_h
> \epsilon h\right\}} \dl{X}_h \leq 1_{\left\{S'_h\leq \epsilon h\right\}}
\dl{X}_h \leq 1_{\left\{S'_h\leq \epsilon h\right\}} \dl{Z}_h(c+\epsilon)
.
\end{eqnarray*}
We have
\begin{eqnarray*}
\frac{1}{h}\E\left[1 - e^{-(y-\dl{X}_h)^+}\right]
&\geq & \frac{1}{h}\E\left[1 - e^{-(y-1_{\{S'_h\leq \epsilon h\}}\dl{Z}_h(c+\epsilon))^+}\right]\\
&=& \frac{1}{h}\E\left[ 1_{\left\{S'_h\leq \epsilon h\right\}}\left(1 -e^{-\left(y-\dl{Z}_h(c+\epsilon)\right)^+}\right)\right]\nonumber\\
&=&\Prob\left(\frac{S'_h}{h}\leq \epsilon \right)\frac{1}{h}\E\left[1 - e^{-(y-\dl{Z}_h(c+\epsilon)^+}
\right].
\end{eqnarray*}
The last equality is due to the independence of $S$ and $S'$.
Since $\lim_{h\rightarrow 0} \frac{S'_h}{h} = 0 $ a.s., which implies $\lim_{h\rightarrow
0} \Prob\left( \frac{S'_h}{h}\leq \epsilon \right)=1$,  we have
\begin{eqnarray*}
\liminf_{h\rr 0} \frac{1}{h}\E\left[1 - e^{-(y-\dl{X}_h)^+}\right] \geq -(c+\epsilon)\mathbf{1}_{\{y=0\}}\mathbf{1}_{\{c+\epsilon
<  0\}}+\Pi(-y).
\end{eqnarray*}
Let $\epsilon \downarrow 0$ in the above inequality, we obtain
\begin{eqnarray*}
\liminf_{h\rr 0} \frac{1}{h}\E\left[1 - e^{-(y-\dl{X}_h)^+}\right] \geq -c\mathbf{1}_{\{y=0\}}\mathbf{1}_{\{c<
 0\}}+\Pi(-y).
\end{eqnarray*}
We have proved (\ref{temp1}). 
\comment{
Recall all the cases above for L\'evy processes with finite variation, we
can conclude 

\begin{eqnarray*}
\lim_{s\rr 0}\frac{1}{s}\E\left[1 - e^{-(y-\dl{X}_s)^+}-1\right]= -c\mathbf{1}_{\{c
< 0\}}\mathbf{1}_{\{y=0\}}+\Pi(-y),
\end{eqnarray*}
where $c$ as the drift can be identified as $$ c=\lim_{s \downarrow 0}\frac{X_s}{s}
\quad a.s.$$  
}
\end{proof}

\begin{rmk}
{\rm Note that if $X$ is  a  L\'evy process with a L\'evy measure $\pi_X$
and $f$ is a bounded continuous function that vanishes in a neighbourhood
of zero, then (\cite[Corollary 8.9]{Sato99})
\begin{equation}\label{temp2}
\lim_{h\downarrow 0}\frac{1}{h}\E[f(X_h)] = \int_{\mathbb{R}}f(x)\pi_X(dx).
\end{equation}
In our case, we aim to compute
\begin{eqnarray*}
\lim_{h\downarrow 0}\frac{1}{h}\E\left[1 - e^{-(y-\dl{X}_h)^+} \right].
\end{eqnarray*}
However, we cannot apply (\ref{temp2}) directly  as $\dl{X}$ is not a L\'evy
process if $X$ is not a monotone process. 
Proposition \ref{prop:Squeezing_positiveDrift} can be viewed as an extension
 for function $f(x)=1 - e^{-(y-x)^+}$ and L\'evy process $X$ with finite
variation.
}
\end{rmk}

\subsection{Indistinguishability of Likelihood Process and Intensity Process}\label{subsec:
compensatorIntensity}
We have proved the existence of the instantaneous likelihood process $\td{\lambda}$
when $X$ is a L\'evy process with finite variation.
Heuristically the intensity process $\lambda$ of the $\G$-compensator should
be equal to $\tilde \lambda$ on the set $\{\tau>t\}$. However, they are not
necessarily the same.

\begin{emp}[\cite{GJZ09}]\rm{
Define a stopping time $\tau:=\inf\{t>0: W_t>y\}$ where $W$ is a Brownian
motion and $y>0$ is a constant. Suppose  $\F$ is the natural filtration of
$W$. We have $\tau$ is a $\F$-stopping time and
\begin{eqnarray*}
\frac{1}{h}\Prob(t<\tau< t+h |\mathcal{F}_t)  &=&\frac{1_{\{\tau>t\}}}{h}\int_0^{h}\frac{|W_t
- y|}{\sqrt{2\pi t^3}}e^{-\frac{(W_t-y)^2}{2t}}dt \stackrel{h\downarrow 0}{\longrightarrow}
0.
\end{eqnarray*}
I.e., $\td{\lambda}_t\equiv 0$ for all $t\geq 0$. As $\tau$ is predictable
under $\F$, the compensator of $N_t=1_{\{\tau \leq t\}}$ is $N_t$, which
indicates the intensity $\lambda$ does not exist. 
}
\end{emp}

Aven's condition in the next lemma provides a sufficient condition that
ensures  $\td{\lambda}$ and $\lambda$ are indistinguishable. 

\comment{
In general, one does not have the information about the existence of intensity
$\lambda$. In this case, in order to investigate $\lambda$, we need to impose
some technical conditions on the tractable instantaneous likelihood  $\td{\lambda}$,
such that the existence and equivalence of $\lambda$ can follow.

\begin{thm}[Aven's Condition\cite{Aven85}]
Let $(h_n)\downarrow 0$, and define
\begin{equation*}
\lambda^n_t:= \frac{1}{h_n}\mathbb{E}\left[N_{t+h_n}-N_t|\mathcal{G}_t\right],
\end{equation*}
where $N$ is an $\G$-adapted point-process.
\comment{
\footnote{Note that for each $t$, $(\lambda^n_t)_{n\geq0}$ is a submartingale.

} 
}

Assume that the following statements hold  with $(\lambda_t)$ and $(d_t)$
non-negative and measurable processes:
\begin{enumerate}
\item for each $t$, $\lim_{n\rightarrow \infty}\lambda^n_t=\lambda_t$ $a.s.$
\item for each $t$ there exists for almost all $\omega$ an $n_0 = n_0(t,\omega)$
such that
\begin{eqnarray*}
|\lambda^n_t(s,\omega)-\lambda(s,\omega)|\leq d(s,w),\quad s \leq t, n\geq
n_0
\end{eqnarray*}
\item $\int_0^t d(s)ds < \infty$ $a.s.$ for each $t$.
\end{enumerate}

Then, $N_t - \int_0^t\lambda_s ds$ is a $\G$-martingale. i.e. $\int_0^t \lambda_s
ds$ is the $\G$-compensator of $N$. 

\end{thm}
}

\begin{lemma}[\cite{Aven85}]\label{Aven's Condition}
If $\lim_{h\rr 0}\lambda^h_t = \td{\lambda}_t$ exists and $\lambda^h_t$ is
uniformly bounded for $t>0$ and $h>0$ $a.s.$, then on $\{\tau>t\}$, $N_t
- \int_0^t \td{\lambda}_s ds$ is a $\G$-martingale, i.e., $\int_0^t \tilde
\lambda_s ds$ is the $\G$-compensator of $N$.
\end{lemma}

With the help of the results of previous  subsections, we can now present the proof of the main theorem.

\noindent{\it Proof of Theorem \ref{thm: main}}. 
 Recall (\ref{k^h}) that on $\{\tau>t\}$
$$
\lambda^h_t = e^{-\dl{X}_t} k^h_t
= \frac{1}{h}  \mathbb{E}\left[1-e^{-\left(y -\dl{X}_{h} \right)^+} \right]\Big|_{y=\dl{X}_{t}-X_t}
$$
and $y\leq 0$, $\dl{X}_{t}\leq 0$ and $c\geq c\wedge 0$, which implies
$\left(y -\dl{X}_{h} \right)^+\leq  -\dl{X}_{h}$  and $\dl{X}_{h}\geq (c\wedge
0)h-S_h$, we have
\begin{eqnarray*}
\lambda^h_t & \leq& \frac{1}{h} \mathbb{E}\left[1-e^{\underline{X}_{h}}\right]
\\
&\leq& \frac{1}{h} \mathbb{E}\left[1-e^{(c\wedge 0)h-S_{h}}\right] \\
&=& \frac{1}{h} \left(1-e^{(c\wedge 0)h-\Pi(0)h} \right)\\
& \leq& -(c\wedge 0)+\Pi(0).
\end{eqnarray*}
Hence the sequence $(\lambda^h_t)_{h>0}$ is uniformly bounded in $t$ and
$h$ a.s., 
Lemma \ref{Aven's Condition} gives the required conclusion that $\lambda$
and $\tilde \lambda$ are indistinguishable on $\{\tau>t\}$, which leads to the expression of $\lambda$ from  Proposition \ref{prop:Squeezing_positiveDrift}.
The proof of  Theorem~\ref{thm: main} is now complete.
\hfill$\Box$

\begin{rmk}{\rm
Similarly, $k^h_t = e^{\dl{X}_t}\lambda^h_t \leq \lambda^h_t$ is also bounded
and with a similar argument as Aven's condition due to the Meyer's Laplacian
approximation theorem, we can conclude that the $\F$-compensator of $\Prob(\tau\leq
t|\Ft_t)$  is $K_t = \int_0^t \td{k}_s ds$ where $\td{k}_t = \lim_{h \rr
0}k^h_t$.
 }
\end{rmk}

\section{Conclusions}\label{sec: conclusion}

In this paper we discuss the intensity problem of a random time 
that is  the first passage time of a finite variation L\'evy process
on a random barrier. We prove the existence of the intensity process and find its explicit representation.  
We  compute the instantaneous credit spread  process explicitly  and give a numerical example for a variance gamma process to illustrate the relation between the credit spread  and 
the distance of the asset value to its running minimal value. 
We thus reconcile the structural model with incomplete information and the path-dependent intensity model  in this setup.

\comment{
\bibliographystyle{plain}       
\bibliography{bib_ACC_Levy}             
}

\bigskip
\noindent{\bf Acknowledgements}. The authors thank the anonymous reviewer and the AE for their comments and suggestions that have helped to improve the previous version. Xin Dong also   thanks Benoit Pham-Dang for useful
discussions.

\end{document}